\documentclass[11pt]{article}
\usepackage{graphicx,amssymb,amsmath}
\usepackage{amsthm}
\usepackage{wrapfig}
\usepackage{psfrag}
\usepackage{algorithm}
\usepackage{algorithmic}
\usepackage{enumerate}
\usepackage{mdwlist} 
\usepackage{comment} 
\usepackage{xspace}
\usepackage[hang,footnotesize,bf]{caption}
\usepackage{color}

\usepackage{thmtools}
\usepackage{thm-restate}

\graphicspath{{./}{figs/}}

\def\Riemann{\mathcal{R}}

\newcommand{\R}{\mathbb{R}}

\newcommand{\C}{\mathbb{C}}

\newcommand{\lfs}{\mathbf{f}}

\newcommand{\dnn}{\mathrm{N}} 
\newcommand{\nn}{\mathop{\mathrm{NN}}} 

\newcommand{\dist}{\mathbf{d}}
\newcommand{\nndist}{\mathbf{d}_{\mathrm{N}}} 
\newcommand{\sqdist}{\mathbf{d}_{\mathrm{sq}}} 
\newcommand{\cdist}{\mathbf{d}_{\mathrm{c}}} 
\newcommand{\nnlen}{\mathbf{\ell}_{\mathrm{N}}}
\newcommand{\sqlen}{\mathbf{\ell}_{\mathrm{sq}}}
\newcommand{\len}{\mathbf{\ell}}

\newcommand{\vor}{\mathrm{Vor}}
\newcommand{\del}{\mathrm{Del}}

\newcommand{\srad}{u}

\newcommand{\corref}[1]{Corollary~\ref{#1}}

\newcommand{\secref}[1]{Section~\ref{#1}}

\newtheorem{theorem}{Theorem}[section]

\newtheorem{lemma}[theorem]{Lemma}

\newtheorem{corollary}{Corollary}
\newtheorem{definition}{Definition}[section]
\newtheorem*{theorem*}{Theorem}
\newtheorem*{lemma*}{Lemma}

\newcommand{\shortversion}[1]{}

\definecolor{gray}{rgb}{0.7,0.7,0.7}

\definecolor{ForestGreen}{rgb}{0.1333,0.5451,0.1333}
\usepackage[letterpaper,
            colorlinks,linkcolor=ForestGreen,citecolor=ForestGreen,
            bookmarks,bookmarksopen,bookmarksnumbered]
           {hyperref}

\usepackage{color}

\title{Approximating Nearest Neighbor Distances\thanks{Partially supported by the NSF grant CCF-1065106.}}

\author{
  Michael B. Cohen\thanks{Massachusetts Institute of Technology, {\tt micohen@mit.edu}}
  \and
  Brittany Terese Fasy\thanks{Department of Computer Science,Tulane University, {\tt brittany.fasy@alumni.duke.edu}}
  \and
  Gary L.~Miller\thanks{Department of Computer Science, Carnegie Mellon University, {\tt glmiller@cs.cmu.edu}}
  \and
  Amir Nayyeri\thanks{School of electrical engineering and computer science, Oregon State University, {\tt nayyeria@eecs.oregonstate.edu}}
 \and
  Donald R. Sheehy\thanks{Computer Science and Engineering Department, University of Connecticut, {\tt don.r.sheehy@gmail.com}}
  \and
  Ameya Velingker\thanks{Department of Computer Science, Carnegie Mellon University, {\tt avelingk@cs.cmu.edu}}
}

\index{Cohen, Michael}
\index{Fasy, Brittany Terese}
\index{Miller, Gary L.}
\index{Nayyeri, Amir}
\index{Sheehy, Donald R.}
\index{Velingker, Ameya}

\begin{document}

\maketitle

\begin{abstract}
Several researchers proposed using  non-Euclidean metrics on point sets in 
Euclidean space for clustering noisy data.  Almost always, a distance function 
is desired that recognizes the closeness of the points in the same cluster, even 
if the Euclidean cluster diameter is large.  Therefore, it is preferred to assign 
smaller costs to the paths that stay close to the input points.

In this paper, we consider the most natural metric with this property, which we 
call the nearest neighbor metric.  Given a point set P and a path $\gamma$, our 
metric charges each point of $\gamma$ with its distance to P.  The total charge 
along $\gamma$ determines its nearest neighbor length, which is formally defined 
as the integral of the distance to the input points along the curve.  
We describe a $(3+\varepsilon)$-approximation algorithm  and a $(1+\varepsilon)$-approximation algorithm to compute the nearest neighbor 
metric.  Both approximation algorithms work in near-linear time.  
The former uses shortest paths on a sparse graph using only the input points.
The latter uses a sparse sample of the ambient space, to find good approximate geodesic paths.
\end{abstract}

\section{Introduction}
\label{sec:intro}
Many problems lie at the interface of computational geometry, machine learning,
and data analysis, including--but not limited to: clustering, manifold learning,
geometric inference, and nonlinear dimensionality reduction.  
Although the input to these problems is often a Euclidean point cloud, a 
different distance measure may be more \emph{intrinsic} to the data.
In particular, we are interested in a distance that recognizes the 
closeness of two points in the
same cluster, even if their Euclidean distance is large, and conversely,
recognizes a large distance between points in different clusters, even if the
Euclidean distance is small.  For example, in Figure~\ref{fig:dbd_0_w_arrows}, the distance between $a$ and $b$ must be larger than the distance between $b$ and $c$.

\begin{figure}[htbp]
  \centering
    \includegraphics[width=0.6\textwidth]{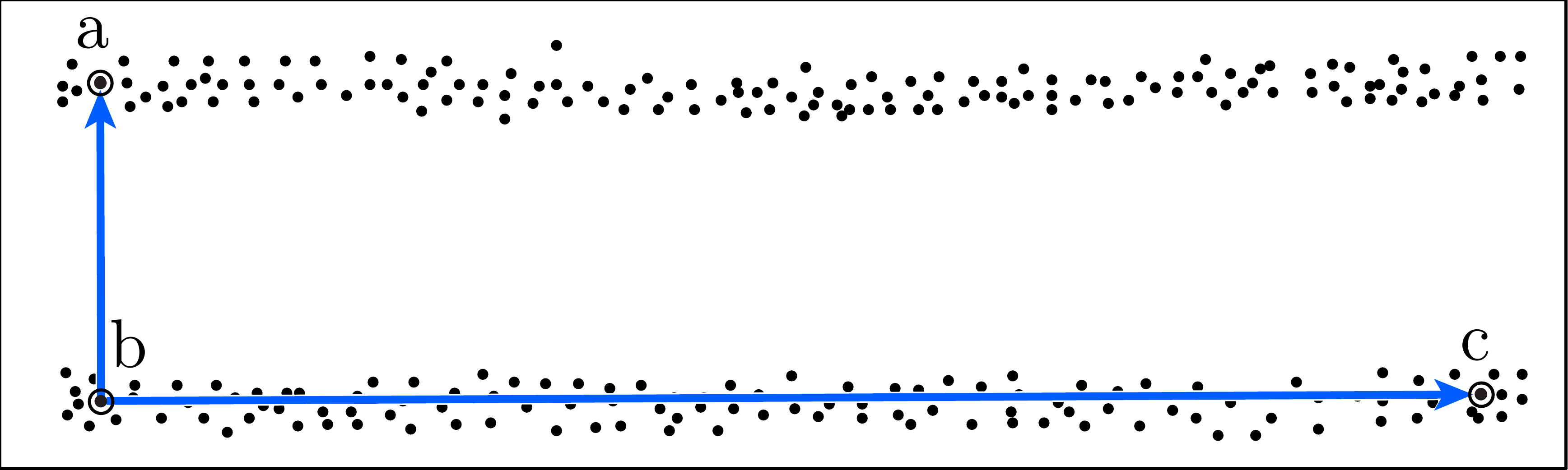}
  \caption{The intrinsic density-based distance can differ from the Euclidean distance.}
  \label{fig:dbd_0_w_arrows}
\end{figure}

There are at least two seemingly different approaches to define a non-Euclidean metric
on a finite set of point on $\R^d$. 
The first approach is to form a graph metric on the point set. An example of
such a graph is the $k$th nearest neighbor graph where a point is only connected to another point
if one is a $k$th nearest neighbor of the other. The edge weights may be a constant  or the Euclidean distances.
In this paper we consider the complete graph but the edge lengths are a power of their Euclidean lengths.
We are particularly interested in the squared length, which we will refer to as the edge-squared metric.

The second approach is to endow all of $\R^d$ with a new metric.
We start with a cost function $c:\R^d \rightarrow \R$, which takes the point cloud into account.
Then, the length of a path $\gamma$ is the integral of the cost function along
the path.
\begin{eqnarray}\label{eqn:costLength}
\ell_c(\gamma) = \int_{\gamma}{c(s)ds} =
\int_0^1{c(\gamma(t))\left|\frac{d\gamma}{dt}(t)\right|dt}.
\end{eqnarray}
The distance between two points $x,y\in \R^d$ is then the length of the shortest path between
them:
\begin{eqnarray}\label{eqn:costDist}
\cdist(x,y) = \inf_{\gamma}{\ell_c(\gamma)} .
\end{eqnarray}
Note that the constant function, $c(x) = 1$ for all $x \in \R^d$, gives the
Euclidean metric; whereas, other functions allow
space to be stretched in various ways.


In almost all applications mentioned above for cost-based metrics, in order to reinforce paths within 
clusters,  one would like to assign smaller lengths to paths that stay close to 
the point cloud.  
Therefore, the simplest natural cost function on $\R^d$ is the distance to the 
point cloud.   
More precisely, given a finite point set $P$ the cost $c(x)$ for $x \in \R^d$ is chosen
to be $\dnn(x)$, the Euclidean distance from $x$ to the nearest point in $P$.
The \emph{nearest neighbor length} ($\dnn$-length) $\ell_\dnn (\gamma)$ of a
curve is given by~\eqref{eqn:costLength}, where we set $c(x) = \dnn(x)$ for
all points $x\in \mathcal{C}$.  We refer to the corresponding metric given
by~\eqref{eqn:costDist} as the \emph{nearest neighbor  metric} or simply the
$\dnn$-distance.


In this paper, we investigate approximation algorithms for $\dnn$-distance computation.  
We describe a $(3+\varepsilon)$-approximation algorithm and a $(1+\varepsilon)$-approximation algorithm.
The former comes from comparing the nearest neighbor metric with the edge-squared metric.
The latter is a tighter approximation that samples the ambient space to find good approximate geodesics.

\subsection{Overview}
In Section~\ref{sec:approximation}, we describe a constant factor approximation algorithm obtained via an elegant reduction into the edge-squared metric introduced by~\cite{bijral11semiSupLearningDBD} and~\cite{vincent04densitySens}. 
This metric is defined 
between pairs of points in $P$ by considering the graph distance on a
complete weighted graph, where the weight of each edge is the square
of its Euclidean length.  
We show that the $\dnn$-distance and edge-squared metric are equivalent up to a factor of three (after a scaling by a factor of four).
As a result, because spanners for the edge-squared metric can be computed in nearly linear time~\cite{lukovszki06resource}, we obtain a $(3+\varepsilon)$-approximation algorithm for computing $\dnn$-distance.

\begin{theorem}
\label{thm:threeApx}
Let $P$ be a set of points in $\R^d$, and let $x,y \in P$.  The nearest neighbor distance between $x$ and $y$ can be approximated within a $(3+\varepsilon)$ factor in $O(n\log n+ n\varepsilon^{-d})$ time, for any $0 < \varepsilon \leq 1$.
\end{theorem}

In Section~\ref{sec:approximation_schemes}, we describe a $(1+\varepsilon)$-approximation algorithm for the
$\dnn$-distance that works in time $\varepsilon^{-O(d)}n\log n$.
Our algorithm computes a discretization of the space for points that are
sufficiently far from $P$.  Nevertheless, the sub-paths that are close to $P$
are computed exactly.
We can adapt our algorithm to work for any Lipschitz cost function that is
bounded away from zero; thus, the algorithm can be applied to many of the
scenarios describe in Appendix~\ref{apx:related}.

\begin{theorem}
\label{thm:ptas}
For any finite set of points $P \subset \R^d$ and any fixed number $0<\varepsilon<
1$, the shortest $\dnn$-distance between any pair of points of the space can be
$(1+\varepsilon)$-approximated in time $O(\varepsilon^{-O(d)}n\log n)$.
\end{theorem}

\section{Related Work}
\label{apx:related}

Computing the distance between a pair of points with respect to a cost function
encompasses several significant problems that have been considered by different
research communities for at least a few centuries.  As early as 1696, Johann
Bernoulli introduced the \emph{brachistochrone} curve, the
shortest path in the presence of gravity, as ``an honest,
challenging problem, whose possible solution will bestow fame
and remain as a lasting monument''~\cite{bernoulli}. With six solutions to his
problem published just one year after it was posed, this event marked the birth
of the field of \emph{calculus of variations}.
In this section, we review some work related to computing shortest paths in a
weighted domain.

\paragraph{Models for Geometries.}  The cost function $c(\cdot)$ in \eqref{eqn:costLength} can
be deliberately picked so that the metric $\cdist(\cdot)$ coincides with
well-known metrics. For example, $c(x,y) = 1$ gives the Euclidean metric in the
plane and $c(x,y) = 1/y$ gives the Poincar\'{e} metric in the half-plane model of
hyperbolic (Lobachevsky) geometry~\cite{kotak10hyperbolic}.

\paragraph{Motion Planning.} Rowe and Ross~\cite{row90gridFreePathPlan}~as well
as Kime and Hespanha~\cite{kim03shortestPathUAV} consider the problem of
computing anisotropic shortest paths on a terrain.   An anisotropic path cost
takes into account the (possibly weighted) length of the path and the direction
of travel.  Note that this problem can be translated into the
problem of computing a shortest path between two compact subspaces of $\R^6$
under a certain cost function.

\paragraph{Computational Geometry.} Indeed, computational geometers are
interested in different versions of this problem.
In the simplest case, $c$ takes values from $\{1, \infty\}$, i.e., the space is
divided into free space and obstacles.
This problem can be solved in polynomial time using visibility graph in two
dimensions, and it can be $\varepsilon$-approximated in three dimensions. 
For example, the computation of the Fr\'echet distance can be posed in this
way~\cite{alt1995computing}.
 A slightly more complicated case occurs when we let $c$ be a piecewise-constant
function.
 Mitchell and Papadimitriou~\cite{mitchell91weightedRegion} formulated this
problem in two dimensions and designed a linear-time algorithm to find the
solution within $\varepsilon$-accuracy.  They list the problem for more general
cost functions as an open problem (See Section 10, problem number (3)).  A
series of
works~\cite{aleksandrov98wrsp,aleksandrov00wrsp,reif00wrsp,aleksandrov05wrsp} has resulted in an $\varepsilon$-approximation algorithm that
computes the shortest paths in a weighted polyhedral surface in $O(
(n/\sqrt{\varepsilon}) \log(n/\varepsilon) \log(1/\varepsilon)$~time.

\paragraph{Machine Learning.} Sajama and Orlitsky~\cite{sajama05estimatingDBDM}
first applied density-based distance (DBDs) to semi-supervised learning.
Assuming that the sample points are taken from an underlying distribution with density $f$, a density-based distance can be defined by setting $c(x) = f(x)^{-1/d}$ in~\eqref{eqn:costLength}~and~\eqref{eqn:costDist}.%
\footnote{The definition of DBDs in~\cite{sajama05estimatingDBDM} is more general in that it allows for a choice of discount functions (we take the inverse raised to the power $1/d$ as this seems to be the most natural way of turning a measure of volumetric density to a measure of length).}
The goal here is to place points that can be connected through dense regions
into the same cluster.
Vincent and Bengio~\cite{vincent04densitySens} and Bousquet et
al.~\cite{bousquet04measurebased} suggest estimating $f$ using a KDE and then
approximating the metric by discretizing the space in a similar fashion to
Tsitsiklis~\cite{tsitsiklis95optimalTrajectories}. However, they do not provide
any analysis on the complexity of the discretized space.
Bijral et al.~\cite{bijral11semiSupLearningDBD} bypass estimating $f$ by
building a complete graph over a set of points $\{x_1, x_2, \dots, x_n\}$
sampled with respect to $f$, in which the length of an edge $(x_i, x_j)$ is
$||x_i-x_j||_p^q$ for fixed $p$ and $q$, and computing pairwise shortest paths
in this graph.
Hwang et al.~\cite{hwang12spthroughsamples} prove that for certain values of $p$ and $q$, the latter metric and the density based metric are equivalent up to a linear factor for sufficiently large values of  $n$.

The nearest neighbor metric can be viewed as a special case of density-based distance when the underlying density is the nearest neighbor density estimator.

\section{Preliminaries}
\label{sec:definitions}

In this section, we define some basic concepts that are used in the paper.


\subsection{Metrics}
In this paper, we consider three metrics.  Each metric is defined by a length function on a set of paths between two points of the space.  The distance between two points is the length of the shortest path between them.

\paragraph{Euclidean metric.}  This is the most natural metric defined by the Euclidean length.  We use $\len(\gamma)$ to denote the Euclidean length of a curve $\gamma$; $\len(\gamma)$ can also be defined by setting $c(x) = 1$ for all $x\in \R^d$ in \eqref{eqn:costLength}.  We use $\dist(x,y)$ to denote the distance between two points $x,y \in \R^d$ based on the Euclidean metric.

\paragraph{Nearest neighbor metric.} As mentioned above, the nearest neighbor length of a curve with respect to a set of points $P$, is defined by setting $c(\cdot)$ to be $\dnn(\cdot)$ in \eqref{eqn:costLength}.  The nearest neighbor length of a curve $\gamma$ is denoted by $\nnlen(\gamma)$, and the distance between two points $x, y \in \R^d$ with respect to the nearest neighbor metric is denoted by $\nndist(x,y)$.

\paragraph{Edge-squared metric.} Finally, the edge-squared metric is defined as the shortest path metric on a complete graph on a point set $P$, where the length of each edge is its Euclidean length squared.  
The length of a path $\gamma$ in this graph is naturally the total length of its edges and it is denoted by $\sqlen(\gamma)$.  The edge-squared distance between two points $x,y \in P$ is the length of the shortest path and is denoted by $\sqdist(x,y)$.

\subsection{Voronoi Diagrams and Delaunay Triangulations}
Let $P$ be a finite set of
points, called \emph{sites}, in $\R^n$, for some $n \geq 1$.
The Delaunay triangulation $\del(P)$ is a decomposition of the plane into
simplices such that for each simplex $\sigma \in \del(P)$,
the Delaunay empty circle property is satisfied; that is, there
exists a circle $C$ such that the vertices of $\sigma$ are on
the boundary of $C$ and int$(C) \cap V$ is empty.  The Voronoi diagram, denoted
$\vor(P)$, is the planar dual to $\del(P)$.
We define the in-ball of a Voronoi cell with site $p$ to be the maximal ball centered at $p$ that is contained in the cell.
The inradius of a Voronoi cell is the radius of its in-ball.
We refer the reader
to~\cite{deberg2000computational} for more details.

\section{Nearest Neighbor Distance Versus Edge-Squared Distance}
\label{sec:approximation}

In this section,  we show that the nearest neighbor distance of two points $x,y 
\in P$ can be approximated within a factor of three by looking at their 
edge-squared distance. More precisely, $\sqdist(x,y)/4 \geq \nndist (x, y) \geq 
\sqdist(x,y)/12$ (see Lemma~\ref{lem:nnupper} and 
Lemma~\ref{lem:lowerConstantBound}).   

As a consequence, a constant factor approximation of the $\dnn$-distance can be obtained via computing shortest paths on a weighted graph, in nearly-quadratic time.
This approximation algorithm becomes more efficient, if the shortest paths are computed on a Euclidean spanner of the points, which is computable in nearly linear time~\cite{Har-peled11Book}.
A result of Lukovszki et al.~(Theorem~16(ii)~of~\cite{lukovszki06resource}) confirms that a $(1+\varepsilon)$-Euclidean spanner is a $(1+\varepsilon)^2$-spanner for the edge squared metric.
Therefore, we obtain Theorem~\ref{thm:threeApx}.

Before, starting the technical part of this section, we remark that both the nearest neighbor and edge-squared metric can have $\Omega(\log n)$ doubling dimension.  An illustrative example is a star with dense points sets on its edges.

\subsection{The Upper Bound}
\label{subsec:nnupper}

We show that the edge-squared distance between any pair of points $x,y \in P$ 
(with respect to the point set $P$) is always larger than four times the 
$\dnn$-distance between $x$ and $y$ (with respect to $P$).
To this end, we consider any shortest path with respect to the edge-squared 
measure and observe that its $\dnn$-length is an upper bound on the 
$\dnn$-distance between its endpoints.  

\begin{restatable}{lemma}{nnupper}
\label{lem:nnupper}
 Let $P = \{p_1, p_2, \dots, p_n\}$ be a set of points in $\R^d$, and let $\nndist$ and $\sqdist$ be the associated nearest neighbor and edge-squared distances, respectively. 
 Then, for any distinct points
 $x,y\in P$, we have that $\nndist(x,y) \leq \frac{1}{4}\sqdist(x,y)$.
\end{restatable}
\begin{proof}
 Consider the shortest path $x = q_1 \to q_2 \to \cdots \to q_k = y$ with respect to the edge-squared metric. 
 Let $\gamma$ be the same path in $\R^d$ parameterized by arc length that uses straight line segments between each pair $q_i$, $q_{i+1}$.
 By the definition of the edge-squared distance, we have
 \begin{eqnarray*}
     \sqdist(x,y) = \sum_{i=1}^{k-1} \dist(q_i, q_{i+1})^2.
 \end{eqnarray*}
On the other hand, by the definition of $\dnn$-distance we have
\[
\nndist(x,y) \leq \ell_\dnn(\gamma) = \int_\gamma \dnn(s)\,ds.
\]
The following sequence of equalities follow by basic rules of integration.
\begin{eqnarray*}
    \int_\gamma \dnn(s)\,ds &=& \sum_{i=1}^{k-1} \int_{q_i\to q_{i+1}}
    \dnn(s)\,ds\\
    &=& \sum_{i=1}^{k-1} \left(\int_{q_i\to\frac{q_i + q_{i+1}}{2}}
    \dnn(s)\,ds + \int_{\frac{q_i + q_{i+1}}{2} \to q_{i+1}}
    \dnn(s)\,ds \right) \\
    &\leq& \sum_{i=1}^{k-1} \left(\int_{q_i\to\frac{q_i + q_{i+1}}{2}}
    \dist(s,q_i)\,ds + \int_{\frac{q_i + q_{i+1}}{2} \to q_{i+1}}
    \dist(s,q_{i+1})\,ds \right) \\
    &=& \sum_{i=1}^{k-1} \left(\int_0^{\dist(q_i, q_{i+1})/2} t\,dt +
    \int_0^{\dist(q_i, q_{i+1})/2} t\,dt \right)\\
    &=& \sum_{i=1}^{k-1} \frac{\dist(q_i, q_{i+1})^2}{4} \\
    &=& \frac{1}{4}\sqdist(x, y).
\end{eqnarray*}
Thus, the proof is complete.
\end{proof}

\subsection{The Lower Bound}
\label{subsec:nnlower}
Next, we show that the edge-squared distance between any pair of points from $P$ cannot be larger than twelve times their $\dnn$-distance.
To this end, we break a shortest path of the $\dnn$-distance into segments in a 
certain manner, and shadow the endpoints of each segment into their closest 
point of $P$ to obtain a short edge-squared path.  The following definition 
formalizes our method of discretizing paths.   

\begin{definition}
Let $P = \{p_1, p_2, \cdots, p_n\}$ be a set of points in $\R^d$, and let $x, y \in
P$. Let $\gamma:[0,1]\rightarrow \R^d$ be an $(x,y)$-path that is internally disjoint from $P$.  A sequence 
$0<t_0\leq t_1\leq \dots\leq t_k < 1$ is a \emph{proper breaking sequence} of~$\gamma$ if it has the following
properties:
\begin{enumerate}
\item The nearest neighbors of $\gamma(t_0)$ and $\gamma(t_k)$ in $P$ are $x$ and $y$, respectively.
\item For all $1\leq i \leq k$, we have $\len(\gamma[t_{i-1}, t_i]) = \frac{1}{2}(\dnn(\gamma(t_{i-1})) + \dnn(\gamma(t_{i})))$
\end{enumerate}
\end{definition}
The following lemma guarantees the existence of breaking sequences.  

\begin{restatable}{lemma}{breakingSeqExists}
\label{lem:breakingSeqExist}
Let $P = \{p_1, p_2, \cdots, p_n\}$ be a set of points in $\R^d$, and let $x, y \in
P$. Let $\gamma$ be a path from $x$ to $y$ that is internally disjoint from $P$.  There 
exists a proper breaking sequence of $\gamma$.
\end{restatable}
\begin{proof}
Pick $t_0$ such that the closest neighbor to $\gamma(t_0)$ in $P$ is $x$.
Inductively, pick $t_i > t_{i-1}$ so that property (2) of a  proper breaking sequence holds until $\gamma(t_i)$'s closest neighbor is $y$.

We need to prove two properties: \textbf{(I)} a $t_i$ with property (2) always exists, \textbf{(II)} this process ends after a finite number of steps (i.e., $t_i$ falls in the Voronoi cell of $y$ for some $i$).

\paragraph{Property I.}  Let $t_{i-1}$ be the last selected point in the process.  We prove that a $t_i \in (t_{i-1}, 1)$ exists such that
\[
\len(\gamma[t_{i-1}, t_i]) = \frac{1}{2}(\dnn(\gamma(t_{i-1})) + \dnn(\gamma(t_{i}))).
\]

Let the functions $f:[t_{i-1}, 1] \rightarrow \R$ and $g:[t_{i-1}, 1]\rightarrow \R$ be defined as follows:
\begin{eqnarray*}
f(t) &=& \len(\gamma[t_{i-1}, t]) \\
g(t) &=& \frac{1}{2}(\dnn(\gamma(t_{i-1})) + \dnn(\gamma(t))).
\end{eqnarray*}
In particular, $f(t_{i-1}) = 0$ and $g(t_{i-1}) = \dnn(\gamma(t_{i-1})) > 0$; so $g(t_{i-1}) > f(t_{i-1})$.
On the other hand,
\begin{eqnarray*}
f(1) = \len(\gamma(t_{i-1}, 1)) \geq \dnn(\gamma(t_{i-1})) > \dnn(\gamma(t_{i-1}))/2 = g(1).
\end{eqnarray*}
The first inequality is a result of the Lipschitz property and the fact that $\dnn(\gamma(1)) = \dnn(y) = 0$.

Since both $f$ and $g$ are continuous functions, the intermediate value theorem implies the existence of a $t_i \in [t_{i-1}, 1]$ such that $f(t_i) = g(t_i)$, which in turn implies property (I).

\paragraph{Property II.}
For $t\in [t_0, 1]$, let $p(t)$ be the Euclidean distance from $\gamma(t)$ to $P\setminus \{y\}$.
By definition, $p(t)$ is positive everywhere.
Since $p$ is continuous and defined on a closed interval, by the extreme value theorem, it attains a minimum value $p_{\min} > 0$.
In any inductive step, if the nearest neighbor of neither $t_{i-1}$ nor $t_i$ is $y$ then $\len(\gamma[t_{i-1}, t_i]) \geq p_{\min}$, by the second property of a breaking sequence.
This implies the existence of a finite $k$ such that $\gamma(t_k)$'s nearest neighbor in $P$ is $y$.
\end{proof}

\begin{lemma}
\label{lem:lipschitzUpperBound}
Let $P = \{p_1, p_2, \cdots, p_n\}$ be a set of points in $\R^d$.  Furthermore, let $\gamma$ be any path in $\R^d$ and $x$ be an endpoint of $\gamma$. 
If $\len(\gamma) = s$, then $\nnlen(\gamma) \geq s \dnn(x) - s^2/2$.
\end{lemma}

\begin{proof}
Let $\gamma_u$ be a unit speed reparameterization of $\gamma$ (i.e., $|\gamma_u(t)| = 1$ for all $t$).
Suppose, without loss of generality, that $x = \gamma(0) = \gamma_u(0)$.
Then, by definition of the nearest neighbor metric,
\[
\nnlen(\gamma) = \nnlen(\gamma_u) = \int_{0}^{s}{\dnn(\gamma_u(t)) dt}.
\]
Then, the Lipschitz property of the $\dnn$ function implies:
\begin{eqnarray*}
\nnlen(\gamma) &\geq&  \int_{0}^{s}{(\dnn(\gamma_u(0))-t) dt} \\
&=& \int_{0}^{s}{(\dnn(x)-t) dt} \\
&=& s \dnn(x) - s^2/2.
\end{eqnarray*}
\end{proof}

Given a path $\gamma$ that realizes the nearest neighbor distance between two points $x$ and $y$, in the proof of the following lemma we show how to obtain another $(x,y)$-path with bounded edge-squared length.  The proof heavily relies on the idea of breaking sequences.

\begin{restatable}{lemma}{lconstbound}
\label{lem:lowerConstantBound}
Let $P = \{p_1, p_2, \cdots, p_n\}$ be a set of points in $\R^d$, and let
$\nndist$ and $\sqdist$ be the associated nearest neighbor and edge-squared
distances, respectively. Then, for any distinct points $x, y \in
P$, $\nndist(x,y) \geq \frac{1}{12} \sqdist(x,y)$.
\end{restatable}
\begin{proof}
Let $\gamma$ be a path from $x$ to $y$ that realizes the nearest neighbor distance between $x$ and $y$, i.e., $\nnlen(\gamma) = \nndist(x,y)$.

Suppose, without loss of generality, that $\gamma$ intersects $P$ only at $\{x, y\}$.
Otherwise, we break $\gamma$ into pieces that are internally disjoint from $P$ and prove the bound for each piece separately.

Let $\{t_0, t_1, \dots, t_k\}$ be a proper breaking sequence of $\gamma$.
For $0\leq i\leq k$, let $n_i \in P$ be the nearest neighbor of $\gamma(t_i)$; in particular $n_0 = x$ and $n_k = y$.

\begin{figure}[tbh]
  \centering
    \includegraphics[height=1.05in]{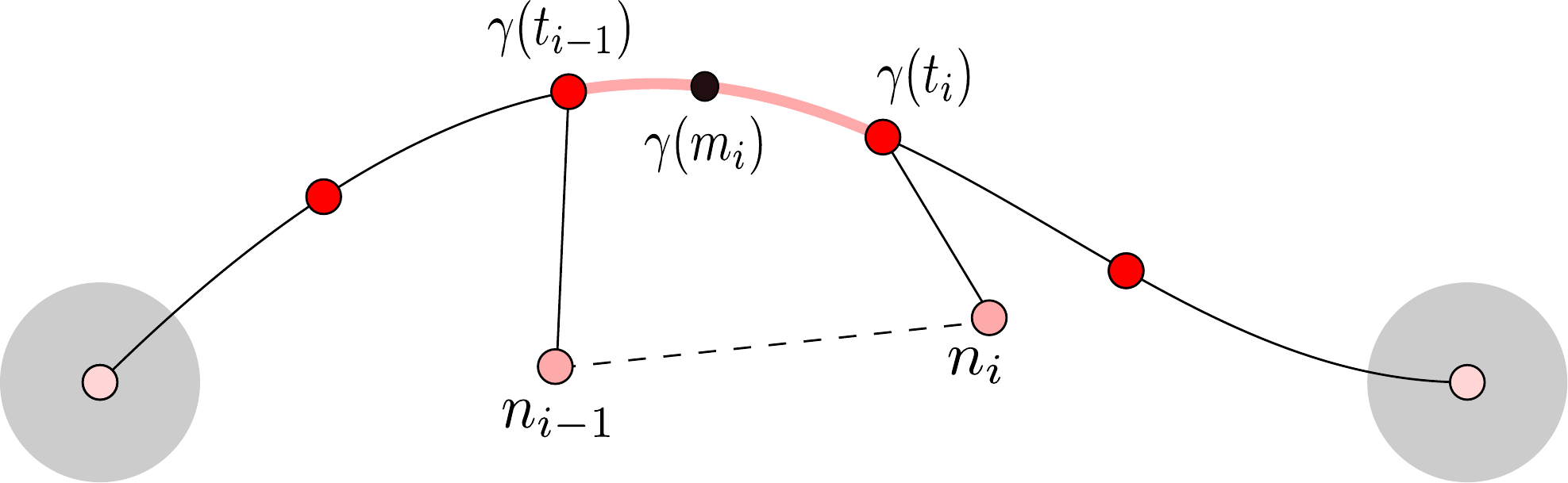}
      \caption{Proof of Lemma~\ref{lem:lowerConstantBound}}
  \label{fig:lowerConstantBound_proof}
\end{figure}

We show that $\nnlen(\gamma[t_{i-1}, t_i]) \geq \frac{1}{12} (\dist(n_{i-1}, n_i))^2$ for any $1\leq i\leq k$, which in turn implies
\[
\nnlen(\gamma) \geq \frac{1}{12}\sum_{i=1}^{k}{(\dist(n_{i-1}, n_i))^2} \geq \frac{1}{12} \sqdist(x,y).
\]

Property~2 of a breaking sequence implies that $\len(\gamma[t_{i-1},t_i]) = 2s$, where $s = \frac{1}{4}(\dnn(\gamma(t_{i-1})) + \dnn(\gamma(t_{i})))$.
We pick $m_i\in(t_{i-1}, t_i)$ so that $\len(\gamma[t_{i-1}, m_i]) = \len(\gamma[m_i, t_i]) = s$.
Lemma~\ref{lem:lipschitzUpperBound} implies
\begin{align}
  \nnlen(\gamma[t_{i-1}, t_i]) 
    &= \nnlen(\gamma[t_{i-1}, m_i]) + \nnlen(\gamma[m_i, t_i])\nonumber\\
    &\geq (s\dnn(\gamma(t_{i-1})) - s^2/2) + (s\dnn(\gamma(t_i)) - s^2/2)\nonumber\\
    &= 3s^2.\label{eqn:nnGammaLBound}
\end{align}

\noindent On the other hand, by the triangle inequality,
\begin{equation}\label{eqn:SqrSegUBound}
\dist(n_{i-1}, n_i) \leq \dnn(\gamma(t_{i-1})) + 2s + \dnn(\gamma(t_{i})) = 6s.
\end{equation}

\noindent The last equality follows by property $I$. Finally, Inequalities~\eqref{eqn:nnGammaLBound} and~\eqref{eqn:SqrSegUBound} imply
\[
\nnlen(\gamma_i) \geq \frac{1}{12}(\dist(n_{i-1}, n_i))^2,
\]
and the proof is complete.
\end{proof}

\section{A ($1+\varepsilon$)-Approximation for the Nearest Neighbor Metric}
\label{sec:approximation_schemes}

In this section, we describe a polynomial time approximation scheme to compute
the $\dnn$-distance between a pair of points from a finite set $P\subset\R^d$.
The running time of our algorithm is $\varepsilon^{-O(d)} n \log n$ for $n$
points in $d$-dimensional space.
We start with Section~\ref{subsec:onecell}, which describes an exact 
algorithm for the simple case in which $P$ consists of just one site. 
Section~\ref{subsec:approxsteiner} describes how to obtain a piecewise 
linear path using infinitely many Steiner points.  
Section~\ref{subsec:approxgraph} combines ideas from 
\ref{subsec:approxsteiner} and \ref{subsec:onecell} to cut down the 
required Steiner points to a finite number. Finally, 
Section~\ref{subsec:steinerconstruct} describes how to generate the necessary 
Steiner points.

\subsection{Nearest Neighbor Distance with One Site}
\label{subsec:onecell}

We describe a method for computing $\nndist$ for the special case that $P$
is a single point using complex analysis. This case will be important
since distances will go to zero at an input point and thus we must be more careful
at input points. Far away from input points, we will use a piecewise constant
approximation for the nearest neighbor function but near input points we will us exact
distances.
More than likely this case has been solved by others since the
solution is so elegant.
We refer the interested reader to~\cite{calVarStrain} for more general methods
to solve similar problems in the field of calculus of variations.

\begin{figure}[h]
  \centering
    \includegraphics[height=.75 in]{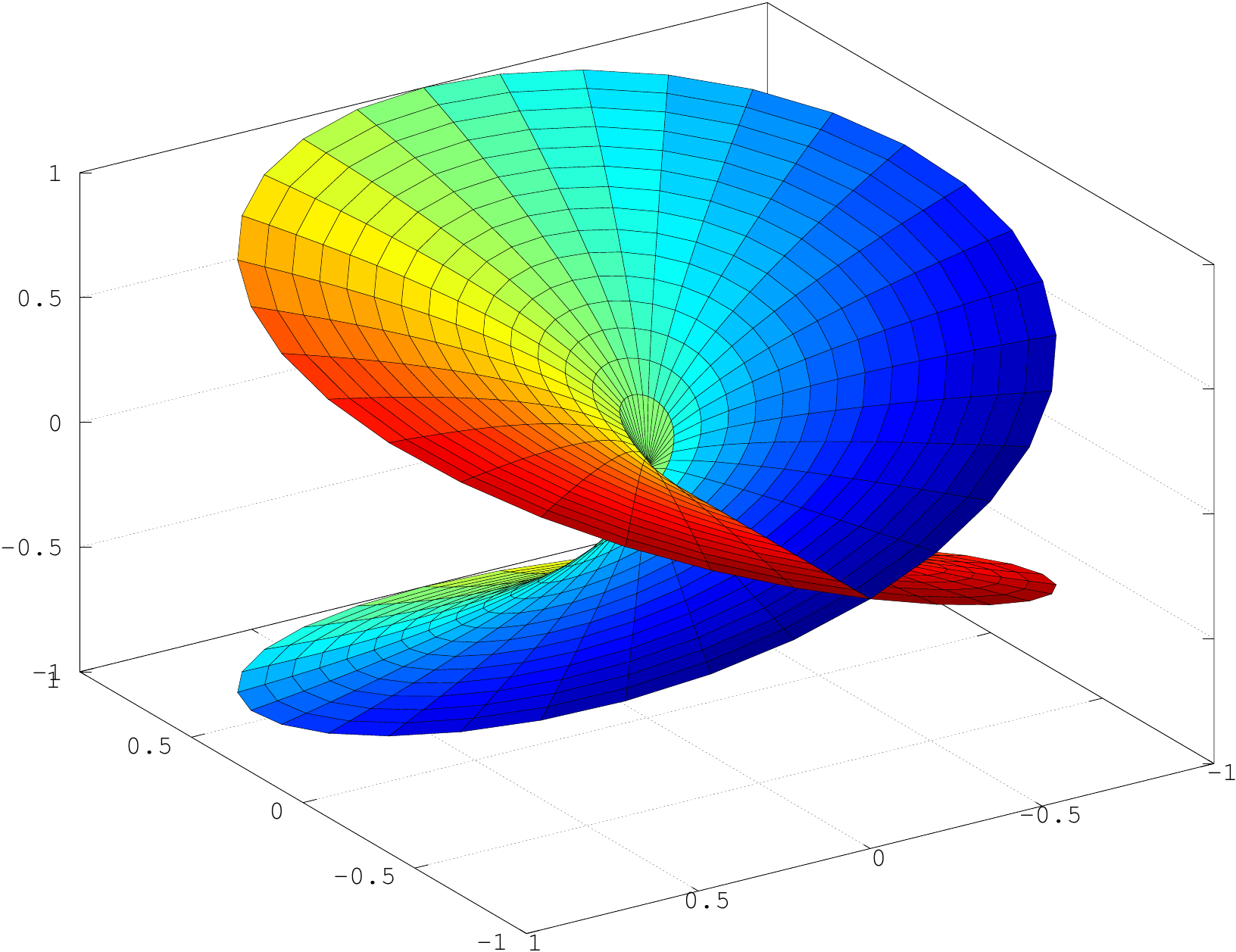}
      \caption{To make the complex function one-to-one, one needs to extend the
complex plane to the two-fold cover
       called the two-fold Riemann Surface.}
  \label{fig:riemann}
\end{figure}

Suppose we want to compute $\nndist(x,y)$ where $P = \{ (0,0) \}$.
Writing $(x,y) \in \C$ in polar coordinates as  $z= r e^{i \theta}$,
we define the \emph{quadratic transformation} $f \colon   \C \to
\Riemann$ by
\begin{equation*}
 f(z) = z^2/2 = (r^2/2) e^{i 2\theta},
\end{equation*}
where $\Riemann$ is the two-fold Riemann
surface; see Figure~\ref{fig:riemann}.
The important point here is
that the image is a double covering of $\C$.
For example, the points $1$ and $-1$ are
mapped to different copies of $1/2$. Therefore,
on the Riemann surface, the distance between $1$ and $-1$ is one and the shortest
path goes through the origin. More generally, given any two nonzero points $p$
and $q$ on the
surface, the minimum angle between them (measured with respect to the origin) will be
between 0 and $2\pi$. Moreover, if this angle is $\geq \pi$, then the shortest
path between them will consist of the two straight lines $[p,0]$ and
$[0,q]$. Otherwise, the line $[p,q]$ will be a line on the surface
and, thus, the geodesic from $p$ to $q$.


Let $\dist_{\Riemann}$ denote the distance on the Riemann surface.
We next show that for a single point, the nearest neighbor geodesic
is identical to the geodesic on the Riemann surface.

\begin{lemma}
 Let $\gamma \colon [0,1] \to \C$ be a curve.  Then,
 the image of $\gamma$ under $f$, denoted by $f \circ \gamma$
 satisfies the following property:
 \begin{equation*}
  \dist_{\Riemann} (f \circ \gamma) = \nnlen (\gamma).
 \end{equation*}
\end{lemma}
\begin{proof}
 Suppose $\gamma \colon [0,1] \to \C$ is any
 piecewise differentiable curve, and let $\alpha :=
 f \circ \gamma$.
 The $\dnn$-length $\nnlen (\gamma)$ of $\gamma$ is
 the finite sum of the $\dnn$-length of all differentiable pieces of $\gamma$. If the
 path $\gamma$ goes through the origin, we further break the path at the origin so
 that $\alpha$ is also differentiable. Thus, it suffices to consider $(a,b) \subset
 [0,1]$ so that $\gamma[a,b]$ is a differentiable piece of $\gamma$. Then, we have
 \begin{align*}
 \nnlen (\gamma[a,b]) & = \int_a^b | \gamma(t)|| \gamma'(t) | \,dt && \mbox{$|
\cdot |$ is modulus.}\\
  & = \int_a^b | \gamma(t)  \gamma'(t) | \,dt  && \mbox{Modulus commutes with
product.}\\
  & = \int_a^b | \alpha'(t) | \,dt  && \mbox{Chain rule.} \\
 & = \len_{\Riemann} (\alpha[f(a), f(b)]).
 \end{align*}
\end{proof}

\begin{corollary}[Reduction to Euclidean Distances on a Riemann Surface]
\label{cor:onecell}
Given three points $x$, $y$, and $p$ in $\R^d$ such that $p = \nn(x) = \nn(y)$,
the nearest neighbor geodesic $G$ from $x$ to $y$ satisfies the following properties:
\begin{enumerate}
\item
$G$ is in the plane determined by $x,y,p$.
\item
\begin{enumerate}
\item If the angle formed by $x,p,y$ is $\pi/2$ or more, then $G$ consists of
the two straight segments $\overline{xp}$ and $\overline{py}$.
\item Otherwise, $G$ is the preimage of the straight line from $f(x)$ to $f(y)$,
where $f$ is the quadratic map in the plane given by $x,y,p$ to the Riemann
Surface.
\end{enumerate}
\end{enumerate}
\end{corollary}


\subsection{Discretizing a Path}
\label{subsec:discretize}

First we approximate any $(x,y)$-path $\gamma$ with a piecewise linear $(x,y)$-path $\alpha$ that is composed of segments that are sufficiently long.
The properties of $\alpha$, formalized in the following definition, are used in the rest of this section to obtain a piecewise linear approximation using Steiner points.

\begin{definition}
Let $P = \{p_1, p_2, \cdots, p_n\}$ be a set of points in $\R^d$, and let $x, y \in
P$. Assume $\gamma$ is a path from $x$ to $y$ that is internally disjoint from $P$.  Also, let $0< \varepsilon_\alpha < 1$ be a real number.  A sequence
$\{0 = t_0, t_1, \dots, t_k, t_{k+1} = 1\}$ is a \emph{$\varepsilon_\alpha$-discretizing sequence} of $\gamma$ if it has the following
properties:
\begin{enumerate}
\item The nearest neighbor of both $\gamma(t_1)$ and $\gamma(t_2)$ is $x$,
while the nearest neighbor of both $\gamma(t_{k-1})$ and $\gamma(t_k)$ is $y$.
\item For all $1\leq i \leq k$, we have $\dist(\gamma(t_{i}), \gamma(t_{i+1})) = \varepsilon_\alpha\cdot \dnn(\gamma(t_{i}))$.
\end{enumerate}
In this case, the piecewise linear path $\alpha = \{x, \gamma(t_1), \dots, \gamma(t_k), y\}$ is a \emph{$\varepsilon_\alpha$-discretized path} of $\gamma$.
\end{definition}

The following lemma guarantees the existence of $\varepsilon_\alpha$-discretizing sequences and, hence, $\varepsilon_\alpha$-discretized paths.
Its proof is essentially the same as the proof of Lemma~\ref{lem:breakingSeqExist}.

\begin{lemma}
Let $P = \{p_1, p_2, \cdots, p_n\}$ be a set of points in $\R^d$, and let $x, y \in
P$. Let $\gamma$ be a path from $x$ to $y$ that is internally disjoint from $P$ and $0<\varepsilon_\alpha<1$ a real number.
There exists a $\varepsilon_\alpha$-discretizing sequence of $\gamma$.
\end{lemma}
Next, we show that the $\dnn$-length of a $\varepsilon_\alpha$-discretized path of $\gamma$ is not much longer than the $\dnn$-length of $\gamma$.

\begin{restatable}{lemma}{plApprox}
\label{lem:plApprox}
If $\alpha$ is a $\varepsilon_\alpha$-discretized path of $\gamma$ for some $\varepsilon_\alpha\leq 1/2$, then
$\nnlen(\alpha) \leq (1+4\varepsilon_\alpha)\nnlen(\gamma)$.
\end{restatable}
\begin{proof}
For $1\leq i\leq k-1$, let $\alpha_i$ be the line segment $(\gamma(t_i),\gamma(t_{i+1}))$.
Furthermore, let $\alpha_0$ and $\alpha_k$ be the line segments $(x,\gamma(t_1))$ and $(\gamma(t_k), y)$, respectively.
Also, let $\gamma_i = \gamma[t_i, t_{i+1}]$, for $1\leq i\leq k-1$, and $\gamma_0 = \gamma[0, t_1]$ and $\gamma_k = \gamma[t_k, 1]$.

We prove the stronger statement that $\nnlen(\alpha_i) \leq (1+4\varepsilon_\alpha)\nnlen(\gamma_i)$, for any $0\leq i\leq k$.
For $i \in \{0, k\}$, the straight line is indeed the shortest path in the nearest neighbor metric (since $\dnn(\gamma(t_1)) = \gamma(t_0) = x$ and $\dnn(\gamma(t_k)) = \gamma(t_{k+1}) = y$), and so, $\nnlen(\alpha_i) \leq \nnlen(\gamma_i)$.
For $1\leq i\leq k-1$, by Lemma~\ref{lem:segUpperBound},
\[
\nnlen(\alpha_i) \leq \nnlen(\gamma_i)\cdot\frac{\dnn(\gamma(t_i))+\len(\alpha_i)}{\dnn(\gamma(t_i))-\len(\alpha_i)}.
\]
Then, by the second property of $\varepsilon_\alpha$-discretizing sequences,
\[
\nnlen(\alpha_i) \leq \nnlen(\gamma_i)\cdot\frac{\dnn(\gamma(t_i))+\varepsilon_\alpha\dnn(\gamma(t_i))}{\dnn(\gamma(t_i))-\varepsilon_\alpha\dnn(\gamma(t_i))}.
\]
After simplifying and using the fact that $\varepsilon_\alpha \leq 1/2$, we obtain
\[
\nnlen(\alpha_i) \leq \nnlen(\gamma_i)\cdot\frac{1+\varepsilon_\alpha}{1-\varepsilon_\alpha} \leq \nnlen(\gamma_i)\cdot(1+4\varepsilon_\alpha).
\]
\end{proof}

\subsection{Approximating with Steiner Points}
\label{subsec:approxsteiner}


Assume $P\subset\R^d$, $x,y\in P$, and let $\gamma$ be an arbitrary $(x,y)$-path.
We show how to approximate $\gamma$ with a piecewise linear path through a collection of Steiner points in $\R^d$.
To obtain an accurate estimation of $\gamma$, we require the Steiner points to be sufficiently dense.
The following definition formalizes this density with a parameter $\delta$.

\begin{definition}[$\delta$-sample]
Let $P = \{p_1, p_2, \cdots, p_n\}$ be a set of points in $\R^d$, and let $D\subseteq \R^d$.
For a real number $0<\delta<1$, a \emph{$\delta$-sample} is a (possibly infinite) set of points $T\subseteq D$ such that if $z \in D\setminus P$, then
$\dist(z, T) \leq \delta\cdot \dnn(z)$.
\end{definition}

The following lemmas guarantees that an accurate estimation of $\gamma$ can be computed using a $\delta$-sample.  

\begin{lemma}
  \label{lem:segUpperBound}
  Let $P = \{p_1, p_2, \cdots, p_n\}$ be a set of points in $\R^d$, and let $x, y \in
  \R^d$. Let $\gamma$ be any path from $x$ to $y$, and $l$ be the straight $(x,y)$-segment.
  Assume further that $\dnn(x) > \len(l)$.
  Then,
  \[
    \nnlen(l) \leq \nnlen(\gamma)\cdot\frac{\dnn(x)+\len(l)}{\dnn(x)-\len(l)}
  \]
\end{lemma}

\begin{proof}
  Let $\gamma' = \gamma \cap B(x, \len(l))$.
  Observe that $\gamma'$ is a finite collection paths, and let $\len(\gamma')$ be the total length of these paths.
  \begin{figure}[htb]
    \centering
      \includegraphics[height=1.25in]{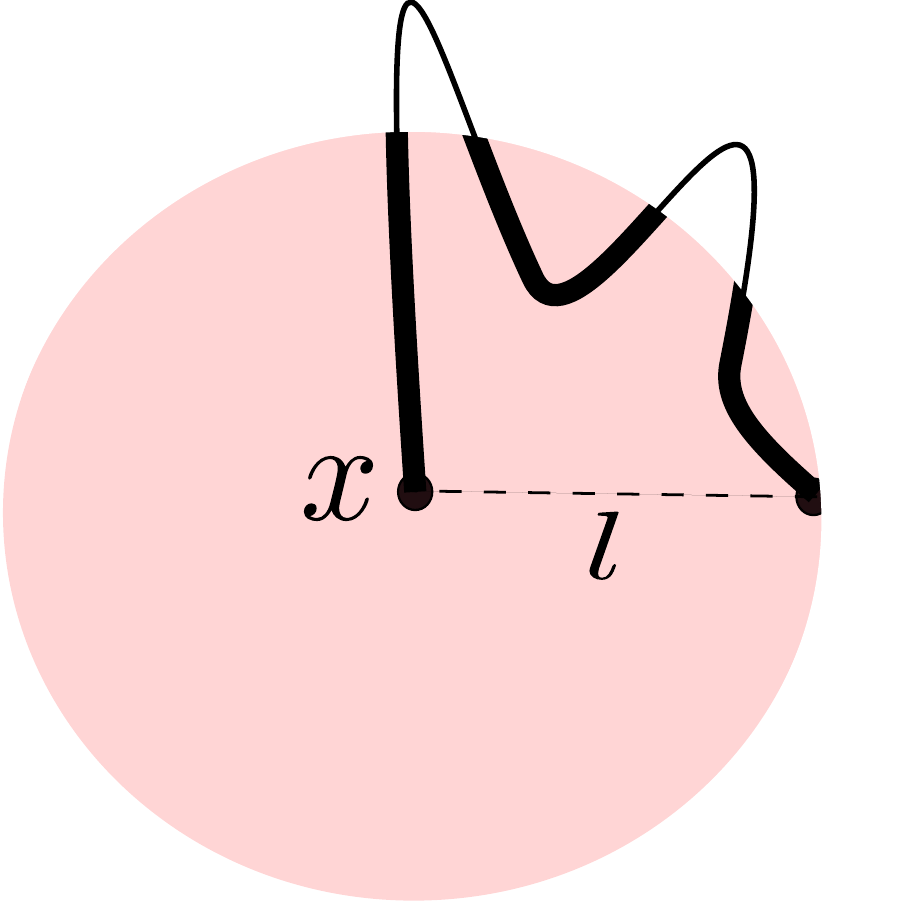}
        \caption{Proof of Lemma~\ref{lem:segUpperBound}; $\gamma'$ is bold.}
    \label{fig:drawing}
  \end{figure}
  For any point $z\in B(x, \len(l))$ we have $\dnn(x) - \len(l) \leq \dnn(z) \leq \dnn(x) + \len(l)$ because $\dnn(\cdot)$ is Lipschitz.
  Since $\gamma'$ and $l$ are both contained in $B(x, \len(l))$, in particular, we have:
  \[
    \nnlen(\gamma') \geq \len(\gamma')(\dnn(x) - \len(l))
  \]
  and
  \[
    \nnlen(l) \leq \len(l)(\dnn(x) + \len(l))\leq \len(\gamma')(\dnn(x) + \len(l)).
  \]
  Thus, since $\gamma' \subseteq \gamma$, we have $\nnlen(\gamma')\leq \nnlen(\gamma)$.  
  Finally, because both sides of the inequalities are positive numbers, we have:
  \[
    \frac{\nnlen(l)}{\nnlen(\gamma)}\leq\frac{\nnlen(l)}{\nnlen(\gamma')}\leq \frac{\dnn(x) + \len(l)}{\dnn(x) - \len(l)}.
  \]
\end{proof}

\begin{lemma}
  \label{lem:SDistGammaDelta}
  Let $P = \{p_1, p_2, \cdots, p_n\}$ be a set of points in $\R^d$, and let $S$ be a $\delta$-sample, and let $0< \delta<1/10$.
  Then, for any pair of points $x,y\in P$, there is a piecewise linear path $\eta = (x, s_1, \dots, s_k, y)$, where $s_1, \dots, s_k \in S$, such that:
  \[
  \nnlen(\eta) \leq (1+C_1\delta^{2/3}) \nndist(x,y),
  \]
  and, for all $1\leq i\leq k-1$,
  \[
  \nnlen((s_i, s_{i+1})) \leq C_2\cdot\delta^{2/3}\cdot\dnn(s_i).
  \]
  $C_1$ and $C_2$ are universal constants.
\end{lemma}
\begin{proof}
Let $\gamma$ be an $(x,y)$-shortest path under the nearest neighbor metric.
We assume, without loss of generality, that $\gamma$ is internally disjoint from $P$.
Of course, we can break any path to such pieces and use induction to prove the lemma for the general case.

Let $\alpha = (x=\gamma(t_0),\gamma(t_1), \dots, \gamma(t_{k+1})=y)$ be a $(2\delta^{2/3})$-discretization of $\gamma$. 
For all $0\leq i\leq k$, let $\alpha_i = (\gamma(t_i), \gamma(t_{i+1}))$ for $0\leq i \leq k$, let $m_i$ be the midpoint of $\alpha_i$, and let $s_i \in S$ be the closest Steiner point to $m_i$. 
Then, let $\beta_i = (\gamma(t_i), s_i, \gamma(t_{i+1}))$ be a two-segment path, and let $\beta = \beta_0 \circ \beta_1 \circ \dots \circ\beta_k$. 
Finally, let $\eta = (x, s_1, s_2, \dots, s_{k-1}, y)$, and let $\eta_i = (s_i, s_{i+1})$ where $0\leq i<k$ and $s_0 = x$ and $s_k = y$ to simplify notation.

We need to prove that $\nnlen(\eta) \leq (1+C_1\delta^{2/3})\nnlen(\gamma)$.
Lemma~\ref{lem:plApprox} implies that
\[
\nnlen(\alpha) \leq (1+8\delta^{2/3})\nnlen(\gamma)
\]
Showing the following two facts imply the statement of the lemma.
\begin{eqnarray}
  \nnlen(\beta) \leq (1+20\delta^{2/3})\nnlen(\alpha)  \label{eqn:betaLEQAlpha} \\
  \nnlen(\eta) \leq (1+8\delta^{2/3})\nnlen(\beta)  \label{eqn:etaLEQBeta}
\end{eqnarray}

\begin{figure}[htb]
  \centering
    \includegraphics[height=1.1in]{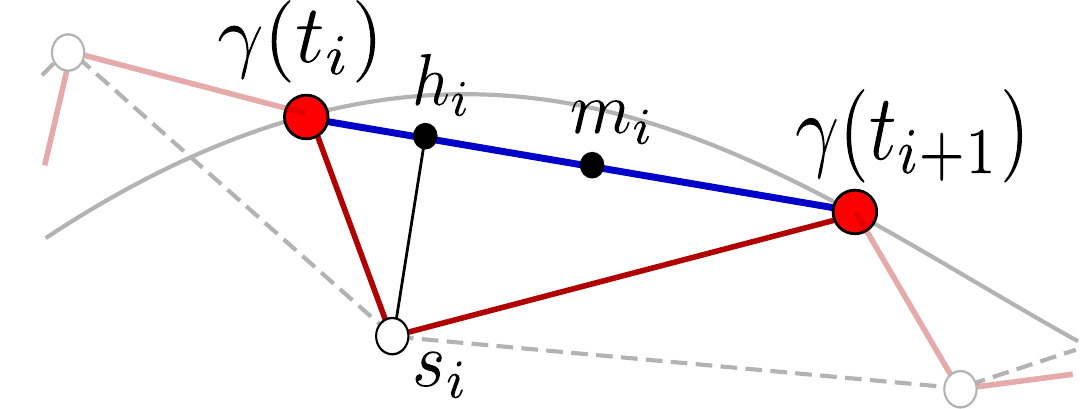}
      \caption{The lines segment $\alpha_i$ is blue, and the two-segment path $\beta_i$ is red.}
  \label{fig:epsApprox}
\end{figure}

\paragraph{A proof of Inequality~\eqref{eqn:betaLEQAlpha}:}
We prove the stronger statement that for all $0\leq i\leq k$, $\nnlen(\beta_i) \leq (1+8\delta^{2/3})\nnlen(\alpha_i)$.
Recall that $m_i$ is the midpoint of $\alpha_i$ and $s_i \in S$ is the closest Steiner point to $m_i$.  Let $h_i$ be the closes point of $\alpha_i$ to $s_i$; see Figure~\ref{fig:epsApprox}.

Because $\len(\alpha_i)\approx \delta^{2/3}\;\dnn(m_i)$, $\len((m_i, s_i))\approx \delta\;\dnn(m_i)$, and $\delta \leq \delta^{2/3}$, we have $(s_i, h_i)$ is orthogonal to $\alpha_i$.
Then, the following bound on the length of the piecewise linear
path $\beta_i = (\gamma(t_i), s_i, \gamma(t_{i+1}))$ is implied
by the Pythagorean theorem using derivative to maximize a function.
\begin{eqnarray}\label{eqn:betaLen}
\len(\beta_i) \leq 2\cdot \sqrt{\left(\frac{\len(\alpha_i)}{2}\right)^2 + (\len(s_i, h_i))^2}\leq
2\cdot \sqrt{\left({\delta^{2/3} \dnn(\gamma(t_i))}\right)^2 + \left(\delta\dnn(m_i)\right)^2}
\end{eqnarray}
Then, the Lipschitz property of the $\dnn$-function implies:
\[
\dnn(m_i) \leq \dnn(\gamma(t_i)) + \frac{\len(\alpha_i)}{2} \leq \dnn(\gamma(t_i)) + {\delta^{2/3}\cdot\dnn(\gamma(t_i))}
\leq (1+\delta^{2/3})\dnn(\gamma(t_i)).
\]
By substituting $(1+\delta^{2/3})\dnn(\gamma(t_i))$ for $\dnn(m_i)$ in \eqref{eqn:betaLen}, we obtain
\begin{eqnarray*}
\len(\beta_i) &\leq&
2\cdot \sqrt{\left({\delta^{2/3} \dnn(\gamma(t_i))}\right)^2 + \left(\delta(1+\delta^{2/3})\dnn(\gamma(t_i))\right)^2} \\
&=&
\delta^{2/3}\dnn(\gamma(t_i))\cdot \sqrt{1 + \left( {\delta^{1/3}(1+\delta^{2/3})} \right)^2} \\
&=&
\len(\alpha_i)\cdot \left(1 + 2\delta^{2/3}\right) \\
\end{eqnarray*}

On the other hand, because both $\alpha_i$ and $\beta_i$ are contained in $B(\gamma(t_i), \varepsilon_\alpha\dnn(\gamma(t_i)) )$, for any point $z \in \alpha_i \cup \beta_i$ we have $\dnn(\gamma(t_i))(1-2\delta^{2/3}) \leq \dnn(z) \leq \dnn(\gamma(t_i))(1+2\delta^{2/3})$ and so
\[
\nnlen(\beta_i) \leq \len(\beta_i)\cdot\dnn(\gamma(t_i))(1+2\delta^{2/3}).
\]
and
\[
\nnlen(\alpha_i) \geq \len(\alpha_i)\cdot\dnn(\gamma(t_i))(1-2\delta^{2/3}).
\]
Last three inequalities together imply
\begin{eqnarray*}
\frac{\nnlen(\beta_i)}{\nnlen(\alpha_i)} &\leq& \frac{1+2\delta^{2/3}}{1-2\delta^{2/3}} \cdot \left(1 + 2\delta^{2/3}\right) \\
&\leq& (1+20\delta^{2/3}).
\end{eqnarray*}
Thus, we obtain Inequality~\eqref{eqn:betaLEQAlpha}
by summing over all $\alpha_i$'s and $\beta_i$'s.

\paragraph{A proof of Inequality~\eqref{eqn:etaLEQBeta}:}
First, by the triangle inequality, we obtain:
\begin{eqnarray*}
\len(\eta_i) &\leq& \len(s_i, t_{i+1}) + \len(t_{i+1}, s_{i+1}) \\
&=& 2\delta^{2/3} \cdot \left(\frac{\dnn(\gamma(t_i))+\dnn(\gamma(t_{i+1}))}{2}\right) \\
&\leq& 2\delta^{2/3} \cdot \max\left(\dnn(\gamma(t_i)),\dnn(\gamma(t_{i+1}))\right)
\end{eqnarray*}
Then, Lemma~\ref{lem:segUpperBound}, letting $x$ be the endpoint with larger $\dnn$ value, implies
\[
\nnlen(\eta_i) \leq \nnlen(s_i, t_{i+1}, s_{i+1})\cdot \frac{1+2\delta^{2/3}}{1-2\delta^{2/3}} \leq \nnlen(s_i, t_{i+1}, s_{i+1})\cdot (1+8\delta^{2/3}).
\]
Finally, by summing over all $\eta_i$'s we obtain Inequality~\ref{eqn:etaLEQBeta}, and the proof is complete.
\end{proof}



\subsection{The Approximation Graph}
\label{subsec:approxgraph}

So far we have shown that any shortest path can be approximated using a $\delta$-sample that is composed of infinitely many points.
In addition, we know how to compute the exact $\dnn$-distance between any pair of points if they reside in the same Voronoi cell of $\vor(P)$.
Here, we combine these two ideas to be able to approximate any shortest path using only a finite number of Steiner points.
The high-level idea is to use the Steiner point approximation while $\gamma$ passes through regions that are far from $P$ and switch to the exact distance computation as soon as $\gamma$ is sufficiently close to one of the points in $P$.

\newcommand{\graphdist}{\mathbf{d}_{\mathcal{A}}}
\newcommand{\graphlen}{\mathbf{\ell}_{\mathcal{A}}}

Let $P = \{p_1, p_2, \cdots, p_n\}$ be a set of points in $\R^d$, and let $B$ be any convex body that contains $P$.
Fix $\delta \in (0,1)$, and for any $1\leq i\leq n$, let $r_i = r_P(p_i)$ be the inradius of the Voronoi cell with site $p_i$. Also, let $\srad_i = (1-\delta^{2/3})r_i$.
Finally, let $S$ be a $\delta$-sample on the domain $B\setminus {\bigcup_{1\leq i\leq n}{B(p_i, \srad_i)}}$.

\begin{definition}[Approximation Graph]
The approximation graph $\mathcal{A} = \mathcal{A}(P, \{\srad_1, \dots, \srad_n\}, S, \delta) = (V_\mathcal{A}, E_\mathcal{A})$ is a weighted undirected graph, with weight function $w:E_\mathcal{A} \rightarrow \R^+$.  The vertices in $V_\mathcal{A}$ are in one to one correspondence with the points in $S\cup P$; for simplicity we use the same notation to refer to corresponding elements in $S\cup P$ and $V_\mathcal{A}$.  The set $E_\mathcal{A}$ is composed of three types of edges:
\begin{enumerate}
\item If $s_1, s_2 \in S$ and $s_1, s_2 \in B(p_i, r_i)$ for any $p_i$, then $(s_1, s_2) \in E_\mathcal{A}$ and $w(s_1, s_2) = \nndist(s_1, s_2)$.  We compute this distance using Corollary~\ref{cor:onecell}.
\item Otherwise, if $s_1, s_2 \in S$ and $\len(s_1, s_2) \leq C_2\delta^{2/3}\max(\dnn(s_1), \dnn(s_2))$, where $C_2$ is the constant of Lemma~\ref{lem:SDistGammaDelta}, then $(s_1, s_2) \in E_{\mathcal{A}}$ and $w(s_1, s_2) = \max(\dnn(s_1), \dnn(s_2))\cdot\len(s_1, s_2)$.
\item If $s_1\in S$ and $s_1\in B(p_i, r_i)$
      then $(p_i, s_1) \in E_\mathcal{A}$ and $w(p_i, s_1) = \nndist(p_i, s_1) = (\dist(p_i, s_1))^2/2$; see \corref{cor:onecell}.
\end{enumerate}
For $x,y \in V_\mathcal{A}$ let $\graphdist(x,y)$ denote the length of the shortest path from $x$ to $y$ in the graph $\mathcal{A}$.
\end{definition}

The following lemma guarantees that the shortest paths in the approximation graph are sufficiently accurate estimations.  

\begin{restatable}{lemma}{approxGraph}
\label{lem:approxGraph}
Let $\{\srad_1, \dots, \srad_n\}$, $S$ and $\delta$ be defined as above.
Let $\mathcal{A}(P, \{\srad_1, \dots, \srad_n\}, S, \delta)$ be the approximation graph for $P$.
For any pair of points $x, y \in P$ we have:
\[
(1-C_2\delta^{2/3})\cdot\nndist(x,y) \leq \graphdist(x,y) \leq (1+C_4\delta^{2/3})\cdot\nndist(x,y),
\]
where $C_2$ and $C_4$ are constants computable in $O(1)$ time.
\end{restatable}
\begin{proof}
First, we prove the upper bound.
Let $D = B\setminus {\bigcup_{1\leq i\leq n}{B(p_i, \srad_i)}}$,  and recall that $S$ is a $\delta$-sample on $D$.
Let $S^+$ be a $\delta$-sample of $B$ such that $S^+ \cap D = S$.
According to Lemma~\ref{lem:SDistGammaDelta}, there is a piecewise linear path $\eta = \{x, s_1, \dots, s_k, y\}$, for $s_1, \dots, s_k \in S^+$, such that $\nnlen(\eta) \leq (1+C_1\delta^{2/3})\nndist(x,y)$.
We use $\eta$ to obtain a $(x,y)$-path $\sigma$ in $\mathcal{A}$ with length at most $(1+2C_2\delta^{2/3})\nnlen(\eta)$; $C_2$ is the constant in Lemma~\ref{lem:SDistGammaDelta}.

Consider any maximal subpath $\eta' = (s_i, s_{i+1}, \dots, s_{j-1}, s_j)$ of $\eta$ that resides in $B\left(p_t, \srad_t/(1-\delta^{2/3})\right)$ for a $p_t \in P$.
The lengths of the type (2) edges and the second inequality of Lemma~\ref{lem:SDistGammaDelta} imply that $s_i, s_j \notin B(p_t, \srad_t)$, which in turn implies that $s_i, s_j \in D$.
Further, $(s_i, s_j)$ is a type (1) edge of $\mathcal{A}$ and $w(s_i, s_j) = \nnlen(s_i, s_j) \leq \nnlen(\eta')$.
So, $\eta'$ in $\eta$ corresponds to a single type (1) edge $(s_i, s_j)$ in $\sigma$ with no longer length.
In fact, all type (1) edges of $\sigma$ correspond to not shorter paths in $\eta$.

Similarly, we consider maximal paths $(p, s_i, s_{i+1}, \dots, s_{j-1}, s_j)$, where $p$ is an input point, and replacing them with type (3) edges of $\sigma$.  Again, all type (3) edges of $\sigma$ correspond to not shorter paths of $\eta$.

It remains to show that type (2) edges of $\sigma$ are not too long.
Observe, that each type (2) edge of $\sigma$ corresponds to a line segment of $\eta$.
Let $(s_1, s_2)\in E_\mathcal{A}$ be a type (2) edge and let $\max(\dnn(s_1), \dnn(s_2)) = \dnn(s_1)$ without loss of generality.
By the Lipschitz property and the length property of type (2) edges
\[
\nnlen(s_1, s_2) \geq (\dnn(s_1) - \len(s_1, s_2))\len(s_1, s_2) \geq (1 - C_2\delta^{2/3})\dnn(s_1)\len(s_1, s_2) = (1 - C_2\delta^{2/3})w(s_1, s_2),
\]
which implies
\[
w(s_1, s_2) \leq (1 + 2C_2\delta^{2/3})\nnlen(s_1, s_2),
\]
and so, by Lemma~\ref{lem:SDistGammaDelta},
\[
\graphlen(\sigma) \leq (1 + 2C_2\delta^{2/3})\nnlen(\eta) \leq (1 + C_4\delta^{2/3})\nndist(x,y).
\]

To prove the lower bound, let $x,y\in V_\mathcal{A}$ and let $\sigma = (x, s_1, \dots, s_k, y)$ be a path from $x$ to $y$ in $\mathcal{A}$.  Also, let $\gamma$ be the $(x,y)$ path in $\R^d$ that hits $(x = s_0, s_1, \dots, s_k, y = s_{k+1})$ in this order, and let the $(s_i, s_j)$ subpath of $\gamma$ be a line segment if $(s_i, s_j)$ is a type (2) edge in $\mathcal{A}$ and be the shortest $(s_i, s_j)$ nearest neighbor path otherwise. Then, let $\gamma_i$ be the subpath of $\gamma$ from $s_i$ to $s_{i+1}$.  We compare $w(s_i, s_{i+1})$ with $\nnlen(\gamma_i)$.

If $(s_i, s_{i+1})$ is type (1) or (3) then $\nnlen(\gamma_i) = w(s_1, s_2)$.  Otherwise, $(s_i, s_{i+1})$ is type (2) and by Lipschitz property
\[
\nnlen(s_1, s_2) \leq (\dnn(s_1) + \len(s_1, s_2))\len(s_1, s_2) \leq (1 + C_2\delta^{2/3})\dnn(s_1)\len(s_1, s_2) = (1 + C_2\delta^{2/3})w(s_1, s_2).
\]
Overall,
\[
\nndist(x, y) \leq \nnlen(\gamma) \leq (1 + C_2\delta^{2/3})\graphlen(\sigma),
\]
and the proof is complete.
\end{proof}

\subsection{Construction of Steiner Points}
\label{subsec:steinerconstruct}

The only remaining piece that we need to obtain an approximation scheme is an algorithm for computing a $\delta$-sample.
For this section, given a point set $T$ and $x\in T$, let $r_T(x)$ denote the inradius of the Voronoi cell of $\vor(T)$ that contains $x$.
Also, given a set $T$ and an arbitrary point $x$ (not necessarily in $T$), let $\lfs_T(x)$ denote the distance from $x$ to its \emph{second} nearest
neighbor in $T$.

We can apply existing algorithms for generating meshes and well-spaced points to compute a $\delta$-sample on $\mathcal{D}\setminus \bigcup_i B(p_i, \srad_i)$, where $\mathcal{D}\subseteq\R^d$ is a domain, and $\srad_i = (1-\delta^{2/3}) r_P(p_i)$.
The procedure consists of two steps:
\begin{enumerate}
\item Use the algorithm of \cite{miller2013WellSpaced} to construct a
well-spaced point set $M$ (along with its associated approximate Delaunay graph) with aspect ratio $\tau$ in
time $2^{O(d)} (n\log n + |M|)$.
\item Then over-refine $M$ to $S$ for the sizing function $g(x) =
\frac{2\delta}{11\tau} \lfs_P(x)$
(while maintaining aspect ratio $\tau$) in time $2^{O(d)} |S|$ by
using the algorithm of Section 3.7 in \cite{SheehyPHD}. (see also~\cite{hudson10topological} for an earlier use of this technique)
\end{enumerate}
In the above algorithm, we will choose $\tau$ to be a fixed constant, say, $\tau = 6$.
Both of the meshing algorithms listed above are chosen for their theoretical guarantees on running time.
In practice, one could use any quality Delaunay meshing algorithm, popular choices include Triangle~\cite{shewchuk96triangle} in $\R^2$ and Tetgen~\cite{siTetGen} or CGAL~\cite{cgalMesh3D} in $\R^3$.

From the guarantees in (\cite{SheehyPHD}), we know that
\begin{eqnarray}
|S| = O\left(\int_\mathcal{D} \frac{dx}{g(x)^d}\right) = \delta^{-O(d)} n\log\Delta, \label{eq:ssize}
\end{eqnarray}
where $\Delta$ is the \emph{spread} of $P$, i.e., the ratio of the largest distance between two points in $P$ to the smallest distance between two points in $P$.

\begin{lemma} \label{lem:lfsnn}
If $x \in \mathcal{D}\setminus B(p_i, \srad_i)$, where $p_i = \nn(x)$, then $\dnn(x) \leq f_P(x) \leq 5 \dnn(x)$.
\end{lemma}
\begin{proof}
Note that $N(x) \leq f_P(x)$ is trivial. To prove the second half of the inequality, let $q$ be the closest point to $p$ in $P\setminus\{p_i\}$. Then, by the triangle inequality, we have that
\begin{eqnarray*}
	f_P(x) &\leq& \dist(x,q)\\
	&\leq& \dist(x,p_i) + \dist(p_i,q)\\
	&\leq& \dist(x,p_i) + 2 r_P(p_i)\\
	&\leq& \dist(x,p_i) + \frac{2}{1-\delta^{2/3}} \dist(x,p_i)\\
	&=& 5 \dist(x,p_i) \\
	&=& 5\dnn(x).
\end{eqnarray*}
\end{proof}

Now, it remains to show that the point set $S$ is indeed a $\delta$-sample on $\mathcal{D}\setminus \bigcup_i B(p_i, \srad_i)$. This is provided by the
following lemma.
\begin{restatable}{lemma}{deltasample}
$S$ is a $\delta$-sample on $\mathcal{D}\setminus\bigcup_i B(p_i, \srad_i)$.
\end{restatable}
\begin{proof}
Let $x\in\mathcal{D}\setminus\bigcup_i B(p_i, u_i)$ be an arbitrary point. Then, let $q$ be the closest point to $x$ that lies in $S$. Note that
\begin{eqnarray*}
\dist(x,q) &\leq& \tau\cdot r_S(q)\\
&\leq& \tau\left(\frac{2\delta}{11\tau}\lfs_P(q)\right)\\
&\leq& \frac{2\delta}{11}(\lfs_P(x) + \dist(x,q)).
\end{eqnarray*}
This implies that
\begin{eqnarray*}
\dist(x,q) &\leq& \frac{2\delta}{11-2\delta} \lfs_P(x)\\
&\leq& \frac{10\delta}{11-2\delta} \dnn(x)\\
&\leq& \delta\dnn(x),
\end{eqnarray*}
as desired.
\end{proof}

Now, we calculate the number of edges that will be present in the approximation graph defined in the previous section. For this, we require
a few lemmas.
\begin{lemma}
Let $A = B(p_i, r_P(p_i))\setminus B(p_i, \srad_i)$ be an annulus around $p_i$. Then, $|A\cap S| = \delta^{-O(d)}$.
\end{lemma}
\begin{proof}
By the meshing guarantees of \cite{SheehyPHD}, we know that for any point $s\in A\cap S$, $B(s, t)$ does not contain a point from $S\setminus\{s\}$ for $t = \Omega(r_S(s)) = \Omega(\delta\cdot r_P(p))$. Thus, the desired result follows using a simple sphere packing argument.
\end{proof}

\begin{lemma}
If $s\in S$, then $|B(s, C_2\delta^{2/3} \dnn(s)) \cap S| = \delta^{-O(d)}$, where $C_2$ is the constant in Lemma~\ref{lem:SDistGammaDelta}.
\end{lemma}
\begin{proof}
As in the previous lemma, meshing guarantees tell us that for any $s'\in B(s, C_2\delta^{2/3} \dnn(s))$, we have that $B(s', t)$ does not contain a point from $S\setminus\{s'\}$ for $t = \Omega(\delta\cdot \dnn(s')) = \Omega(\delta\cdot \dnn(s))$. Thus, we again obtain the desired result from a sphere packing argument.
\end{proof}
From the above lemmas, we see that $\mathcal{A}$ is composed of $|S| = \delta^{-O(d)} n\log\Delta$ vertices and
$n\delta^{-O(d)} + |S|\cdot \delta^{-O(d)} = |S|\cdot\delta^{-O(d)}$ edges.

\noindent\textbf{Remark}.
Note that the right hand side of \eqref{eq:ssize} is in terms of the spread, a non-combinatorial quantity. Indeed, one can construct examples of $P$ for which the integral in \eqref{eq:ssize} is not bounded from above by any function of $n$. However, for many classes of inputs, one can obtain a tighter analysis. In particular, if $P$ satisfies a property known as \emph{well-paced}, one can show that the resulting set $S$ will satisfy $|S| = 2^{O(d)} n$ (see~\cite{miller08linear,sheehy12new}).

In a more general setting (without requiring that $P$ is well-paced), one can modify the algorithms to produce output in the form of a \emph{hierarchical mesh}~\cite{miller11beating}.
This then produces an output of size $2^{O(d)} n$, and $(1+\varepsilon)$-approximation algorithm for the nearest neighbor metric can be suitably modified so that the underlying approximation graph uses a hierarchical set of points instead of a full $\delta$-sample.
However, we ignore the details here for the sake of simplicity of exposition.

The above remark, along with the edge count of $\mathcal{A}$ and the running time guarantees from~\cite{miller2013WellSpaced}, yields Theorem~\ref{thm:ptas}, the main theorem of this section.
 
\section{Discussion}
\label{sec:questions}

Motivated by estimating geodesic distances within subsets of $\R^n$, we consider
two distance metrics in this paper: the $N$-distance and the edge-squared
distance.  The main focus of this paper is to find an approximation of the
$N$-distance.  One possible drawback of our $(1+\varepsilon)$-approximation algorithm is its exponential dependency on $d$.
To alleviate this dependency a natural approach is using a Johnson-Lindenstrauss type projection.
Thereby, we would like to ask which properties are preserved under random projections 
such as those in Johnson-Lindenstrauss transforms.


We are currently working on implementing the approximation algorithm presented
in \secref{sec:approximation}.  We hope to show that this approximation is
fast in practice as well as in theory.

\section*{Acknowledgement}
The authors would like to thank Larry Wasserman for helpful discussions.
\bibliographystyle{alpha}
\bibliography{../bibliography}

\clearpage

\end{document}